\documentclass[runningheads]{llncs}

\usepackage{graphicx}
\usepackage[T1]{fontenc}

\usepackage[english]{babel} 
\usepackage{hyperref}
\usepackage{ulem}

\usepackage{amsmath,amssymb} 
\usepackage{mathtools}

\usepackage{color,xcolor}
\usepackage{subcaption} 
\usepackage{algorithm,algpseudocode}
\usepackage{varioref}
\usepackage{cancel}
\usepackage{pgfplots}
\usepackage{tikz,pgf}
\pgfplotsset{compat=1.16}
\usetikzlibrary{arrows}
\usetikzlibrary{positioning}

\newtheorem{thm}{Theorem}[section]
\newtheorem{corr}[thm]{Corollary}
\newtheorem{lem}[thm]{Lemma}
\newtheorem{defn}[thm]{Definition}
\newtheorem{rem}[thm]{Remark}
\newtheorem{nota}[thm]{Notation}

\def \N {\mathbb{N}}
\def \Z {\mathbb{Z}}
\def\Z{\mathbb Z}
\def\Q{\mathbb Q}
\def\R{\mathbb R}
\def\C{\mathbb C}
\def\OO{\mathcal{O}}

\def\H{\mathcal H}

\DeclareMathOperator{\Vol}{\mathrm{Vol}}

\DeclareMathOperator{\Aut}{\mathrm{Aut}}
\DeclareMathOperator{\m}{\mathfrak{M}}
\DeclareMathOperator{\Mat}{\mathrm{Mat}}

\DeclareMathOperator{\Log}{\mathrm{Log}}
\DeclareMathOperator{\BDD}{\mathrm{BDD}}
\DeclareMathOperator{\CVP}{\mathrm{CVP}}
\DeclareMathOperator{\Lip}{\mathrm{Lip}}
\DeclareMathOperator{\GL}{\mathrm{GL}}

\DeclareMathOperator{\disc}{\mathrm{disc}}

\DeclareMathOperator{\Prob}{\mathrm{Prob}}
\DeclareMathOperator{\dist}{\mathrm{dist}}
\DeclareMathOperator{\vect}{\mathrm{vect}}
\DeclareMathOperator{\poly}{\mathcal{P}}
\DeclareMathOperator{\polylog}{\mathrm{polylog}}
\DeclareMathOperator{\str}{\mathrm{str}}
\newcommand{\norm}[1]{\left\lVert#1\right\rVert}

\algrenewcommand\algorithmicrequire{\textbf{Input:}}
\algrenewcommand\algorithmicensure{\textbf{Output:}}

\begin{document}

\title{The special case of cyclotomic fields in quantum algorithms for unit groups}

\iftrue
\author{Razvan Barbulescu\inst{1}\thanks{The first author has been funded by the Hybrid quantum initiative (HQI) of the France 2030 program.} and Adrien Poulalion\inst{2}}

\institute{
 Univ. Bordeaux, CNRS, Bordeaux INP, IMB, UMR 5251,  F-33400, Talence, France\\  \url{razvan.barbulescu@u-bordeaux.fr}
\and
Alice \& Bob, Corps des Mines, France\\
\url{adrien.poulalion@mines.org}
}

\else
\author{Undisclosed authors}
\fi

\date{}
\maketitle
\begin{abstract}
Unit group computations are a cryptographic primitive for which one has a fast quantum algorithm, but the required number of qubits is $\tilde{O}(m^5)$. In this work we propose a modification of the algorithm for which the number of qubits is $\tilde{O}(m^2)$  in the case of cyclotomic fields.  Moreover, under a recent conjecture on the size of the class group of $\Q(\zeta_m+\zeta_m^{-1})$, the quantum algorithms is much simpler because it is a hidden subgroup problem (HSP) algorithm rather than its error estimation counterpart: continuous hidden subgroup problem (CHSP). We also discuss the (minor) speed-up obtained when exploiting Galois automorphisms thnaks to the Buchmann-Pohst algorithm over $\OO_K$-lattices.
\end{abstract}

\section{Introduction}

 The difficulty to compute class groups and its cardinality, the class group, plays an important role in cryptography. Notably, they are at the foundation of a time commitments protocol~\cite{Castagnos2021commitement}, a scheme of homomorphic encryption~\cite{Castagnos2015homomorphic} and a verifiable delay function~\cite{Wesolowski2019VDF}. 
Note that the particular cases play an important role as the former two examples use quadratic fields and for the latter the cyclotomic case can be the fastest. 

From a perspective of theoretical computer science, these schemes are broken because there exists a quantum algorithm of polynomial time complexity~\cite{biasse2016efficient}. However, from a cryptographic point of view it is required to obtain a more precise estimation of the amount of qubits and of quantum gates, as it was done for the other primitives of public key cryptography
\cite{Bernstein2019,Bernstein2021concrete,Liu2021,Beauregard2003,Proos2003,Zalka2006,Haner2017,Lauter2017quantum,Liu2021}.

Let us recall the chronology of the works addressing this question. Note that both in the classical and quantum algorithms class group and unit group are very similar and they are sometimes computed simultaneously in order to have a halting condition, as in the classical algorithm of Buchmann and  McCurley. In 1994, Shor designed a quantum algorithm to factor integers and, respectively, solve the discrete log problem in cyclic groups \cite{shor1994algorithms}. Kitaev~\cite{kitaev1995quantum}  reformulated the algorithms by reducing both problems to finding the set of periods of functions defined over $\Z^r$ with finite $r$ and taking values in a finite set; this is the hidden subgroup problem~(HSP). Note that the parameter $r$ affects the time and space complexity only in a polynomial manner.

In 2002, Hallgren~\cite{hallgren2007polynomial}\footnote{The version of 2002 had 6 pages whereas the version published in 2007 has 19 pages.} reduced the computation of a fundamental unit of a real quadratic field to finding the set of periods of a function defined over $\R$. More generally, if a function defined over $\R^m$ has a lattice (discrete subgroup) of periods, then finding the set of periods is the  continuous hidden subgroup problem~(CHSP). As for factoring and discrete log, the unit group is sub-exponential on a classical computer and polynomial on a quantum~one.

When $m=1$, the main difference between HSP and CHSP is the problem of finding $\alpha\in\R$ when given approximations of $k \alpha$ and $\ell \alpha$ for two integers $k$ and~$\ell$ ; this is solved using continued fractions. To compute fundamental units for a family of number fields of constant degree, one has to solve CHSP for $m\geq 1$ bounded by a constant. In contemporary works, Hallgren~\cite{hallgren2005fast} and Schmidt~\cite{schmidt2005polynomial} replaced the continued fractions by an LLL-based algorithm of Buchmann and Pohst~\cite{BuchmannPohst}. 

To this point, CHSP and HSP have an identical quantum part and differ in the classical post-treatment. When the degree of the number fields is free to be unbounded in a family of discriminants going to infinity, one has to solve CHSP for unbounded parameters $m$, possibly as large as the logarithm of the discriminants. In this case, the previous algorithms~\cite{hallgren2005fast,schmidt2005polynomial} require an exponential time in $m$, the degree of the number fields, as it was shown by Biasse and Song in \cite[Prop. B.2]{biasse2019cyclo}\footnote{The proposition has number $2$ in the version of 2015.}.

In 2014, Eisenträger et al.~\cite{eisentrager2014quantum-short}\footnote{The 10-page-long original version of 2014 doesn't contain the proofs, which were made public only in the 47-page-long version of 2019.} achieved to make this algorithm polynomial-time for arbitrary degree number fields. For this the periodic function is multiplied by a Gaussian. Before the full version~\cite{eisentrager2014quantum} of this work was made public, a second team, de Boer et al.~\cite{boer2020quantum}, worked on a more thorough analysis of the algorithm and  established in 2019 the precise complexities (space and time) for the CHSP. 

\paragraph{Our contribution and cryptographic recommendations.}
The goal of this article is to make possible a precise resource comparison between computing unit and class groups on one side and breaking symmetric cryptosystems like AES on the other side, as requested by the NIST specifications~\cite{NIST}. Indeed, experts in the technology of quantum computers study the NISQ scenario in which error-corrected quantum computers with 100-to-1000 qubits become available whereas a quantum computer with $10^{20}$ qubits remains unfeasible or too expensive to be considered. In this context, the estimation of the number of qubits as ``polynomial'' in~\cite{eisentrager2014quantum} is not enough. We instantiate the case of unit and class group in the general frame of CHSP ~\cite{boer2020quantum} and obtain that the number of qubits is $O(m^5)$, e.g. when $m=10000$ a quantum attack requires $10^{20}$ error-corrected qubits. 

We continue by investigating possible weak number fields $K$: is it a security problem if $K$ has automorphisms? Is it a weakness if $K$ is cyclotomic? 

In Section~\ref{sec:automorphisms} we explain that automorphisms allow to implement LLL over $\OO_K$ and obtain a speed-up which is at most polynomial in the number of automorphisms. 

In the case of cyclotomic fields, the speed-up is not automatic: if one runs the general algorithm on these fields, the number of qubits is $O(m^5\log m)$. However, we make the following  observation: 
\textbf{Assume that $L$ is an unknown lattice,  $M\subset L$ is known and has a short basis, and one has an algorithm to produce vectors near $L^*$. Then one can use the short basis of $M$ to bring any vector near $L^*$ to their nearest vector of $M^*$ and they will automatically be in $L^*$.}

\underline{Cryptographic recommendations.} A precise estimation of the resources to compute class and unit groups shows that they are more resistant than symmetric cryptosystems if one avoids weak keys. If one wants to use class group based (CGP-based) cryptography then:
\begin{itemize}
\item one must avoid fixed degree fields as they use $O(m)$ qubits;
\item one must avoid cyclotomic fields as they use $O(m^2\log m)$ qubits; Also a HSP-based algorithm exists;
\item one might prefer to use fields with automorphisms to speed-up the computations as they require $\Omega(m^3)$ qubits and the speed-up concerns only the Buchmann-Pohst algorithm and is polynomial in the number of automorphisms.
\end{itemize}

\paragraph{Road map.} The article is organized as follows. In Section~\ref{sec:improve}
, we explain our improvement as a general problem of lattices, outside the context of unit groups and quantum computing.  Then, in Section~\ref{sec:previous} we make a detailed presentation of the quantum algorithms for unit groups and combine results to state the complexity of the algorithms for arbitrary number fields. We exploit the automorphisms in Subsection~\ref{sec:automorphisms}. In Section~\ref{sec:cyclotomic}
 we instantiate the previous complexities in the case of cyclotomic and abelian fields, and compare with our improvement from Section~\ref{sec:improve}. Finally, in Section~\ref{sec:HSP} we make an attempt to reduce the unit-group computations to HSP instead of CHSP, this reduces massively the number of qubits needed, especially for cyclotomic fields of index $p^k$ for small primes~$p$.

\section{Our improvement seen as a lattice problem}\label{sec:improve}

Let $\langle x, y\rangle$ denote the dot product of two vectors $x,y\in \C^n$. Given a lattice $L\subset \R^m$, we call dual of $L$ the lattice 
\begin{equation*}
L^*=\{y\in \R^m\mid \forall x\in L, \langle x, y\rangle \in \Z^n\}.
\end{equation*}
Let us recall a series of properties of $L^*$ (see~\cite{banaszczyk1993new} for a reference).
\begin{lem}\label{lem:dual}
\begin{enumerate}
\item If $L$ is generated by the rows of a matrix $B$ then $L^*$ is generated by the rows of $(B^{t})^{-1}$; in particular $\det L^*=1/\det L$;
\item If $M$ is a sublattice of $L$ then $L^*\subset M^*$ and $[L:M]=[M^*:L^*]$. 
\end{enumerate}
\end{lem}

\subsection{Informal presentation}
The CHSP algorithm to compute unit groups calls $L$ the lattice of units and follows the following strategy:
\begin{enumerate}
\item Use a quantum procedure to generate vectors of $\C^n$ in a small ball around each vector of a set of generators of $L^*$.
\item Apply the Buchmann-Pohst algorithm (which is classical of polynomial time) and the previously found vectors to find a basis of $L^*$. 
\item Invert the matrix of the previously found basis to obtain a basis of $L$ (this is classical of polynomial time).
\end{enumerate}

A characteristic of Buchmann-Pohst's algorithm is that it has a large precision decrease\footnote{See Th 10 in~\cite{boer2020quantum} for a precise estimation.}, namely one needs a large number of bits of precision on the vectors computed in Step $1$ in order to have much fewer bits of precision on the basis of $L^*$. Given a precision requirement $\tau$, i.e. $\log \tau$ bits of precision, in~\cite{boer2020quantum} one computes that the input of Buchmann-Pohst must have precision which depends on~$k$. In this work we reduce the number of qubits required by Step 1.

In the case of cyclotomic fields, one has a basis for the lattice of cyclotomic units $M\subset L$. This allows to compute a basis of $M^*$ which contains the unknown lattice $L^*$. Our idea is as follows: when given a vector sufficiently close to $L^*$, solve CVP with respect to $M^*$ and automatically it will be in $L^*$. In more detail, let $m_1,\ldots,m_n$ be a basis of $L$. We do Step 1 at low precision and obtain $\widetilde{y}$ and then we correct it into $y$ so that $\langle y,m_i\rangle \in \Z^n$ for all $i=1,2,\ldots,n$. At this point we have a generating set of $L^*$ at high precision. For simplicity of exposition we can still apply the Buchmann-Pohst algorithm. In view of a practical implementation one can also use an algorithm for Hermite normal form.

\begin{figure}
\definecolor{qqqqff}{rgb}{0,0,1}
\definecolor{ududff}{rgb}{0.30196078431372547,0.30196078431372547,1}
\definecolor{cqcqcq}{rgb}{0.7529411764705882,0.7529411764705882,0.7529411764705882}
\begin{center}
\begin{tikzpicture}[line cap=round,line join=round,>=triangle 45,x=1cm,y=1cm,scale=0.8]
\draw [color=cqcqcq, xstep=1cm,ystep=1cm] (-5.77,-3.24) grid (3.03,4.24);
\clip(-5.77,-3.24) rectangle (3.83,5.24);
\draw (1.21,-0.04) node[anchor=north west] {CVP};
\draw [line width=2pt] (1,-1) circle (0.544885309033011cm);
\draw (-3.11,3.76) node[anchor=north west] {$M^*$};
\draw [color=qqqqff](-3.07,2.88) node[anchor=north west] {$L^*$};
\begin{scriptsize}
\draw [fill=ududff] (-2,-1) circle (4.5pt);
\draw [fill=ududff] (1,-1) circle (4.5pt);
\draw [fill=ududff] (-1,1) circle (4.5pt);
\draw [fill=ududff] (2,1) circle (4.5pt);
\draw [fill=ududff] (-4,1) circle (4.5pt);
\draw [fill=ududff] (-5,-1) circle (4.5pt);
\draw [fill=ududff] (-3,-3) circle (4.5pt);
\draw [fill=ududff] (-6,-3) circle (4.5pt);
\draw [fill=ududff] (0,-3) circle (4.5pt);
\draw [fill=black] (-3,0) circle (1pt);
\draw [fill=black] (-2,0) circle (1pt);
\draw [fill=black] (-2,1) circle (1pt);
\draw [fill=black] (-3,1) circle (1pt);
\draw [fill=black] (-4,0) circle (1pt);
\draw [fill=black] (-5,0) circle (1pt);
\draw [fill=black] (-5,1) circle (1pt);
\draw [fill=black] (-6,1) circle (1pt);
\draw [fill=black] (-6,0) circle (1pt);
\draw [fill=black] (-6,-1) circle (1pt);
\draw [fill=black] (-6,-2) circle (1pt);
\draw [fill=black] (-5,-3) circle (1pt);
\draw [fill=black] (-4,-3) circle (1pt);
\draw [fill=black] (-4,-2) circle (1pt);
\draw [fill=black] (-4,-1) circle (1pt);
\draw [fill=black] (-5,-2) circle (1pt);
\draw [fill=black] (-3,-2) circle (1pt);
\draw [fill=black] (-3,-1) circle (1pt);
\draw [fill=black] (-1,0) circle (1pt);
\draw [fill=black] (-1,-1) circle (1pt);
\draw [fill=black] (-1,-2) circle (1pt);
\draw [fill=black] (-2,-2) circle (1pt);
\draw [fill=black] (-2,-3) circle (1pt);
\draw [fill=black] (-1,-3) circle (1pt);
\draw [fill=black] (0,1) circle (1pt);
\draw [fill=black] (0,0) circle (1pt);
\draw [fill=black] (0,-1) circle (1pt);
\draw [fill=black] (1,1) circle (1pt);
\draw [fill=black] (1,0) circle (1pt);
\draw [fill=black] (2,0) circle (1pt);
\draw [fill=black] (2,-1) circle (1pt);
\draw [fill=black] (2,-2) circle (1pt);
\draw [fill=black] (2,-3) circle (1pt);
\draw [fill=black] (1,-3) circle (1pt);
\draw [fill=black] (1,-2) circle (1pt);
\draw [fill=black] (0,-2) circle (1pt);
\draw [fill=black] (-2,3) circle (1pt);
\draw [fill=black] (-1.51,3.02) circle (1pt);
\draw [fill=black] (-1,3) circle (1pt);
\draw [fill=ududff] (-2.07,2.34) circle (4.5pt);
\draw [fill=ududff] (-1.51,2.32) circle (4.5pt);
\draw [fill=ududff] (-0.97,2.3) circle (4.5pt);
\end{scriptsize}
\end{tikzpicture}
\end{center}
\caption{Illustration of Algorithm~\ref{algo:Babai}.
}
\label{fig:illustration}
\end{figure}
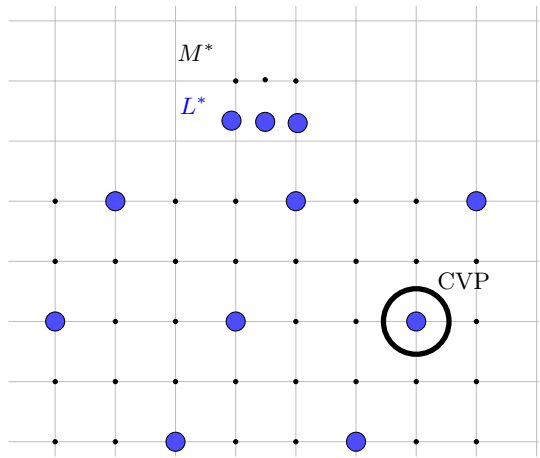

\subsection{Precise statement}
Step 1 produces vectors $\widetilde{y}$ close to $L^*$ such that $y:=\CVP(\widetilde{y},L^*)$ follow a continuous distribution on $\R^m$ which is close to a discrete Gaussian distribution $c$ on $L^*$. The following definition is a precise description of its output and focuses on the properties of $c$ used in the following sections.
\begin{defn}[Dual lattice sampler]\label{defn:dual lattice sampler}
Let $c:L^*\rightarrow \C$ be a map such that $\sum_{\ell^*\in L^*} |c_{\ell^*}|^2=1$. Let $1/4>\epsilon>0$, $1/2>\delta>0$ and $r>0$ be three parameters. 
An algorithm is a dual lattice sampler of parameters $(\epsilon,\delta,r)$ if it outputs a vector $x \in \R^m$ such that, for any finite set $S\subset L^*$, one has
\begin{equation*}
    \Prob\left(y\in \bigcup_{\ell^*\in S} B(\ell^*, \delta \lambda_1^*)\right)\geq \sum_{\ell^*\in S} |c_{\ell^*}|^2 - \eta.
\end{equation*}
for a map $c$ satisfying:

\begin{enumerate}
    \item \textbf{Uniformity property}: there exists $\varepsilon \leq 1/4$ such that, for every strict sublattice $N \subsetneq L^*$ : $$ \sum_{\ell^* \in N} |c_{\ell^*}|^2 < \frac{1}{2}+\varepsilon .$$
    \item \textbf{Concentration property}: There exists $r=r(m)$ and $0<p<\frac{1}{2}-\varepsilon-\eta$ such that: $$ \sum_{|\ell^*| > r} |c_{\ell^*}|^2 < p. $$
\end{enumerate}
\end{defn}

We can now define the problem we tackle. 

\begin{problem}\label{problem:dual}
Let $M\subset L\subset \R^m$ be two lattices. Let $\epsilon,\delta,r>0$ and $k\in \N$ be parameters. We are given
\begin{itemize}
\item a basis of $M$;
\item an upper bound on $[L:M]$;
\item a dual lattice solver for $L$ of parameters $\delta$ and $\eta$ of our choice.
\end{itemize}
Compute a basis $(z_1,\ldots,z_m) \in L$ given by the value of $[L:M]$ and the coordinates of the $z_i$ in a basis of $\frac{1}{[L:M]}M$. 
\end{problem}
Let $B_M$ be the matrix whose rows are the basis $(z_1,\ldots,z_m)$ of $M$. We propose the following solution to this problem, where we will make precise later the value of $\delta$ so that the time complexity is polynomial.

\begin{algorithm}
\caption{Full computation of $L$ using the dual lattice sampler.}
\begin{algorithmic}[1]
\Require 
\begin{itemize}
\item upper bound of $\det L$ of the lattice~$L$ and a lower bound on $\lambda_1^*$
\item $\epsilon,\delta, r>0$ and a dual sampler of $L$ of these parameters; 
\item a basis of a sublattice $M\subset L$
\end{itemize}
\Ensure a basis of $L$
\State $k\gets m  \log_2(\sqrt{m}\Lip(f)) +\log_2(\det L)) $  \Comment{Value from \cite[Th 3]{boer2020quantum}}
\For{$i=1,2,\ldots,k$}\Comment{\textcolor{blue}{Step 1 - Quantum}}
\State $\widetilde{y_i} \gets \text{output}( \text{dual lattice sampler, \fbox{$\delta\lambda_1^*(L)=1/2\norm{B_M}$}})$ 
\State \fbox{$y_i\gets \BDD(\widetilde{y_i},M^*,1/2\norm{B_M})$} \label{task BDD}
\EndFor 

\State \sout{Use Buchmann-Pohst algorithm on $(y_1,\ldots,y_t)$.} Compute the Hermite normal form to find a basis $(y_1',\ldots,y'_m)$ of $L^*$ from $(y_1,\ldots,y_k)$. Here the vectors have exact integer coordinates in a basis of $M^*$. \Comment{\textcolor{blue}{Step 2 - Classical} }\label{task HNF}
\State Compute the Smith normal form (SNF) of $(y_1',\ldots,y'_m)$ to obtain the exact value of $[L:M]$. Output $(B^{-1})^t\in \frac{1}{[L:M]}\Mat_m(\Z)$, where $B\in \Mat_m(\Z)$ is the matrix a basis of $L^*$ written in a basis of $M^*$. \Comment{\textcolor{blue}{Step 3 - Classical}}
\end{algorithmic}
\label{algo:full L}
\end{algorithm}

\begin{nota}\label{notation:norm}
For any matrix $B\in \GL_n(\C)$ and any $\alpha,\beta\in \R_+\bigcup \{\infty\}$ the operator norm is  
$$\norm{B}_{\alpha,\beta} = \max\{ \norm{Mv}_\alpha \mid \norm{v}_\beta\leq 1\}.$$
When $\alpha$ and $\beta$ are not specified\footnote{The norm operator used in the analysis of the Buchmann-Pohst algorithm~\cite[Appendix B]{boer2020quantum} is $\norm{\cdot}_{\infty,1}$.}, $ \norm{B}$ denotes $\norm{B}_{\infty,1}$:
$$\norm{B}=\norm{B}_{\infty,1} = \max_j \sum_i |B_{i,j}|.$$ 
When $\alpha=\beta$ we simply write $\norm{B}_\alpha$. Note that $$\norm{B}_{\infty,1}\geq \norm{B}_{2}=\max_i\sqrt{\sum_j |B_{i,j}|^2}\geq \norm{B}_{\infty,1}/\sqrt{m} .$$
In this article, the complexity of the algorithms depends on $O(\log \norm{B}_{\infty,1})=O(\log \norm{B}_{2} +\log m)$. Since $\log \norm{B}=\Omega(m)$, this justifies that we drop the indices of the operator norm.  
\end{nota}

\begin{thm}\label{th:correctness}
a) Algorithm~\ref{algo:full L} solves Problem~\ref{problem:dual} if $\delta \lambda_1^*(M) < 1/(2 \norm{B_M})$. The number of qubits is $O(Qm)$ with $Q$ given in Equation~\eqref{lem:qubits}.
\end{thm}
We prepare the proof with two lemmas.

\begin{lem}[Lemma $5$ in \cite{boer2020quantum}]\label{lem:value of k}
We note $k=\alpha (m+m\log_2 R+\log_2(\det L))$, for an absolute constant $\alpha>1$. 

Let $\widetilde{y_1},\widetilde{y_2},\ldots,\widetilde{y_k}$ be the first $k$ vectors output by a dual basis sampler. For $i=1,2,\ldots,k$ put $y_i=\CVP(\widetilde{y_i},L)$. Then for any value of the absolute constant $\alpha >2$ we have
$$\Prob( y_1,\ldots,y_k \text{ generate }L) \geq 1-c^{m},$$
where $c<1$ is an explicitly computable constant.
\end{lem}

From now on, we fix such a $k$, and we use it in Algorithm~\ref{algo:full L}. 

\begin{algorithm}
\begin{algorithmic}
\Require $\widetilde{y}\in \R^m$ such that $d(\widetilde{y},M^*)< 1/(2\norm{B_M})$
\Ensure $\CVP(\widetilde{y},M^*)$, as a vector and as coordinates in a basis of $M^*$
\State compute $\widetilde{z}:=B_M^t\widetilde{y}$;
\State round $z=(z_1,\ldots,z_n):=(\lfloor\widetilde{z_1}\rceil,\ldots,\lfloor\widetilde{z_n}\rceil $)
\State \Return $y:=B_{M^*}z\in M^*$ and $z\in \Z^m$.
\end{algorithmic}
\caption{Babai's $\BDD$ solver.}
\label{algo:Babai}
\end{algorithm}

Let us put $\delta$ such that $d(\widetilde{y},M^*)=\delta \lambda_1^*(M)$. The following result can be seen as a reformulation of a result of Babai.
\begin{lem}[\cite{Babai86} Eq. (4.3)] \label{lem:Babai}
Algorithm~\ref{algo:Babai} solves $\BDD(\widetilde{y},M^*,\delta\lambda_1^*(M))$ in classical polynomial time when $\delta \lambda_1^*(M)< 1/(2\norm{B_M}_2)$, where $B_M$ is the matrix of a basis of~$M$. 
\end{lem}
\begin{proof}

We claim that $\lfloor B_M^t \widetilde{y}\rceil = B_M^t y'$ where $y'=\CVP(\widetilde{y},M^*)$. Indeed 
\begin{small}
$$\norm{B_M^t\widetilde{y}-B_M^ty'}_\infty\leq \norm{B_M^t\widetilde{y}-B_M^ty'}_2\leq \norm{B_M}_2\cdot \norm{\widetilde{y}-y'}_2=\norm{B_M}_2\delta\lambda_1^*(M) <1/2.$$ 
\end{small}
Since $ \lfloor B_M^t \widetilde{y}\rceil = B_M^t y' $, by Lemma~\ref{lem:dual}, $y=B_{M^*}=((B_M)^t)^{-1}(B_M^ty')=y'$, so the algorithm is correct.
\end{proof}

\begin{proof}[Proof of Theorem~\ref{th:correctness}.a)]
We structure the proof in several steps.

\underline{step 1.} We apply Lemma~\ref{lem:value of k} to the lattice $L^*$. Then, with probability $1-c^m$, the vectors $y_1,y_2,\ldots,y_k$  generate $L^*$.

\underline{step 2.} 
Let us fix $i\in \{1,2,\ldots,k\}$. By the definition of the dual sampler, with probability greater than $1-\eta$, $\dist(\widetilde{y_i},L^*)< \delta \lambda_1^*(M^*)$.  
Hence Lemma~\ref{lem:Babai} applies and Line~\ref{task BDD} is executed in classical polynomial time. 

Since, $L^*\subset M^*$ and $y_i':=\CVP(\widetilde{y_i},L^*)$ belongs to $M^*$. Since $d(\widetilde{y_i},L^*)<\lambda_1^*(M)/2 $, $y_i'$ is the closest vector of $M^*$ to $\widetilde{y_i}$, or equivalently $y_i'=y_i$. We summarize by $y_i=\CVP(\widetilde{y},L)$.

\underline{step 3.} Algorithm~\ref{algo:Babai} can output $z$ in addition to $y$. This means that the vectors $y_i\in L^*$ are given by their integer coordinates with respect to a basis of $M^*$. 

In Line~\ref{task HNF}, one computes the HNF~\cite[Sec 2.4.2]{cohen1996course} by manipulating integer coefficients only. This is done in classical polynomial time without loss of precision.

\underline{step 4.} Since $B_M\in \Mat_{k}(\Z)$, $(B_M^{-1})\in \frac{1}{\det B_M}\Mat_k(\Z)$. Finally, $B_M$ and its SNF~\cite[Sec. 2.4.4]{cohen1996course} have the same determinant, which equals $[L:M]$. Hence, the last step computes $[L:M]$ in polynomial time.
\end{proof}

\section{Previous results on the quantum unit-group calculation}\label{sec:previous}

\subsection{A summary on CHSP}

We first define the Continuous Hidden Subgroup Problem to which the problem of computing the unit group can be reduced. 

\begin{defn}[Continuous Hidden Subgroup Problem - CHSP \cite{eisentrager2014quantum}]\label{defn:CHSP} Let $f~:~ \R^m \rightarrow \mathcal{S}$, where $\mathcal{S}=\oplus_{i\in\{0,1\}^n} \C |i\rangle$ is the space of states of $n$ qubits. \\
The function $f$ is \textbf{an $(a,r,\varepsilon)$-oracle hiding the full-rank lattice} $L$ if and only if it verifies the following technical conditions:
\begin{enumerate}
    \item $L$ is the period of $f$, i.e. $\forall x\forall \ell\in L, f(x+\ell)=f(x)$. (periodicity)
    \item The function $f$ is $a$-Lipschitz. (Lipschitz condition)
    \item $\forall x,y\in \R^m$ such that $dist(x-y,L) \geq r$, we have $|\langle f(x)~\vert ~f(y) \rangle| \leq \varepsilon$. (strong periodicity)
\end{enumerate}
 Given an efficient quantum algorithm to compute~$f$, compute the hidden lattice of periods~$L$. 
\end{defn}

Let us now recall the algorithm which solves CHSP. We mimick the technique which solve the Hidden Subgroup Problem (HSP). For this latter, let $g~:~ G \rightarrow \mathcal{S}$, with $G$ a finite cyclic  group. The solution goes as follows:
\begin{enumerate}
\item compute the superposition $\sum_{x\in G}|x\rangle$ and then apply $g$ to obtain the superposition of the values of $g$:~ $|\psi(g)\rangle=\sum_{x \in G} g(x)\vert ~ x ~\rangle$;
\item apply the Quantum Fourier Transform (see~\cite{shor1994algorithms}) to obtain\\ $\sum_{y \in G} \hat{g}(y)\vert ~ y~\rangle$; this expression is equal to $\sum_{y \in H^\perp} \hat{g}(y)\vert ~ y~\rangle$
\item measure the state to obtain an element $g\in H^\perp$.
\end{enumerate}
A few iterations of the algorithm allow to have a set of generators of $H^\perp$ and hence to compute $h$. 

Coming back to the continuous case, if we were able to implement the Fourier Transform, we would be able to draw random vectors of $L^*$. If a set of generators is known, then one can extract a basis to completely describe $L^*$. Then $L$ would be obtained thanks to Lemma~\ref{lem:dual}. The obstacle is here that $\R^m$ is infinite and the Fourier Transform (QFT) is now an integral instead of a finite sum so we cannot compute the QFT precisely. As a way around, we do an approximate computation of the QFT by the means of a Riemann sum, which amounts to restrict the domain to a segment and to discretize~it. Set $\rho_{\sigma}(x)~=~ \exp(- \frac{\pi^2 || x||^2 }{\sigma^2})$.

\begin{lem}[Thm. 2 in~\cite{boer2020quantum}]\label{lem:dual lattice sampler}
There exists dual lattice sampler quantum algorithm of polynomial time which uses $Qm+n$ qubits, where 
\begin{equation}\label{lem:qubits}
Q=O\left(m\log\left(m\log \frac{1}{\eta}\right)\right)+ O\left( \log \left(\frac{\Lip(f)}{\eta\delta\lambda_1^*}  \right)\right).
\end{equation}
\end{lem}

\begin{algorithm}
\caption{A CHSP solver}
\label{algo:CHSP}
\begin{algorithmic}[1]
\Require A function $f$ which can be computed in polynomial time, as in Definition~\ref{defn:CHSP} having a hidden period lattice $L$. We require $\det L$ up to one bit of precision and $\Lip(f)$. A parameter $\tau$ of required precision.
\Ensure $(\tilde{x_1},\ldots, \tilde{x_m})$ such that $(x_1,\ldots,x_m)$ is a basis of~$L$ for $x_i=\CVP(\tilde{x_i})$ and $\norm{x_i-\tilde{x_i}}\leq \tau$ 
\State $k\mapsto \log_2(\sqrt{m}\Lip(f)) +\log_2(\det L))$ \Comment{value from \cite[th 3]{boer2020quantum}}  
\For{$i=1,2,\ldots, k$ }  \Comment{\textcolor{blue}{Step 1 - Quantum}}
\State $\tilde {y_i}\gets  \text{output}$(dual lattice sampler: \fbox{$\delta \lambda_1^*=\frac{(\lambda_1^*)^3(\det L)^{-1}}{2^{O(mk)}\norm{B_{L^*}}_\infty^m}$})\Comment{value from \cite[Appendix B.3 Th 4]{boer2020quantum}}
\State pass \Comment{compare to Algorithm~\ref{algo:full L}}
\EndFor 

\State Apply the Buchmann-Pohst algorithm on $(\tilde{y_1},\ldots,\tilde{y_k})$ and obtain $(\tilde{y_1'},\ldots,\tilde{y_m'})$ so that $\CVP(\tilde{y_i},L^*)_{i=1,m}$ is a basis of $L^*$   \Comment{\textcolor{blue}{Step 2 - Classical} }
\item Output the rows of $(B^{-1})^t$ where $B$ has rows $\tilde{u'_i}$.  \Comment{\textcolor{blue}{Step 3 - Classical}}
\end{algorithmic}
\end{algorithm}
\begin{thm}[Space complexity of the CHSP solver, Th 1 of~\cite{boer2020quantum}]\label{th:CHSP}
We set\footnote{The complexity of the CHSP solver depends on $\log 1/\eta$ and not on $\eta$ itself. Hence the contribution of $\eta$ is hidden in the $\tilde{O}$ notation.} $\eta=1/k^2$. Let $N_f$ be the amount of qubits necessary to encode $f$. We fix $\tau$ the error expected on basis' vectors. Then, Algorithm~\ref{algo:CHSP} is correct and requires a number of qubits:

\begin{equation}\label{eq:qubits}
\begin{array}{rcl}
N_{\text{qubits}}&=&O(m^3\cdot \log(m))+O(m^3 \cdot \log(\Lip(f)))+O(m^2\cdot \log(\mathrm{det}(L)))\\
&&\qquad +O \left (m \cdot \log \frac{\Lip(f)}{\lambda_1^*}  \right )+O \left (m \cdot \log \frac{1}{\lambda_1^* \cdot \tau} \right )+O(N_f).
\end{array}
\end{equation}
\end{thm}

The complexity of the algorithm depends on $\lambda_1^*$ which has the advantage that it is an invariant of the lattice. However, $\lambda_1^*$ is difficult to lower bound and in some applications one has bounds on the coefficients of $B_L$, the matrix of some basis of $L$. 

\begin{lem}
$ 2^{-3m}\norm{B_L^t}    \leq  \frac{1}{\lambda_1^*}\leq \norm{B_L}$.
\end{lem}
\begin{proof}
The first inequality is~\cite[Corollary 6 in Appendix B.1]{boer2020quantum}.

The second inequality is direct. Let $v\in L^*$ such that $\norm{v}_2=\lambda_1^*$. By Notation~\ref{notation:norm} we have $\norm{B_Lv}_1\leq \norm{B_L}\norm{v}_\infty$. Since $v\in L^*$, $v\neq 0$, $B_Lv\in \Z^n$ and then $\norm{B_Lv}_1\geq 1$. Also, $\norm{v}_2\geq \norm{v}_\infty$, so $$1\leq \norm{B_L}\lambda_1^*,$$  
which proves the second inequality.
\end{proof}

\subsection{Reduction of the computation of the unit group to CHSP}

In this subsection, we present the method of \cite{eisentrager2014quantum} to reduce the computation of the unit group to the CHSP. For completeness, we reproduce the description of~\cite[Sec ]{biasse2016efficient}. Let $G:= \R^{n_1+n_2}\times (\Z/2\Z)^{n_1}\times (\R/\Z)^{n_2}$, and the mapping 

$$
\begin{array}{c}
\varphi:~~ (u_1,\ldots,u_{n_1+n_2},\mu_1,\ldots,\mu_n,\theta_1,\ldots,\theta_{n_2})\\
\mapsto ((-1)^{\mu_1}e^{u_1},\ldots,(-1)^{\mu_{n_1}}e^{u_{n_1}},e^{2i\pi\theta_1}e^{u_{n_1+1}},\ldots,
e^{2i\pi\theta_{n_2}}e^{u_{n_1+n_2}})
\end{array}
$$

Let $\sigma_1,\sigma_{n_1}$ and respectively $\sigma_{n_1+1},\ldots,\sigma_{n_1+n_2}$ be the real and the complex embeddings of $K$. We consider the Cartesian ring $$E=\sigma_1(\OO_K)\times \cdots \times \sigma_{n_1+n_2}(\OO_K).$$ An $E$-ideal is a sub-$\Z$-lattice $\Lambda$ of $E$ which is such that $\forall x\in E, x\Lambda \subset \Lambda$.

We define
$$ 
\begin{array}{cccl}
g:&G& \rightarrow & \{E\text{-ideals}\} \\
& x & \mapsto &  g(x)E.
\end{array}
$$

\begin{defn}[Ex. 5.3 of~\cite{eisentrager2014quantum}]
Set $\H=\otimes_{i\in \N}\C|i\rangle $, $\H_Q= \otimes_{i=0}^Q\C|i\rangle  $ and $\pi_Q:\H\rightarrow \H_Q$ the projection on the first $Q$ qubits. We define the straddle encodings of parameter $\nu$, $\str_\nu:\R\rightarrow \H$, $\str_ {m,\nu}:\R^m\rightarrow \H^m $ and $f:\{\text{lattices of }\R^m\}\rightarrow \H^m$ as follows:
\begin{itemize}
\item $|\str_\nu(x)\rangle =\cos(\frac{\pi}{2}t)|k\rangle+\sin(\frac{\pi}{2}t)|k+1\rangle$, where $k=\lfloor x/\nu\rfloor$, $t=x/\nu-k$; 
\item $|\str_{m,\nu}(x_1,\ldots,x_m)\rangle = \otimes_{i=1}^m  |\str_\nu(x_i)\rangle$;
\item $|\str_{\text{lattice},m,\nu}(L)\rangle=\gamma^{-1/2}\sum_{x\in L}e^{-\pi\norm{x}^2/s}|\str_{m,\nu}(x)\rangle $\\ with $\gamma=\sum_{x\in L}e^{-2\pi\norm{x}^2/s^2}$.
\end{itemize}
Note that computing $\str_{\text{lattice},m,\nu}$ up to an error $\tau$ requires $O(m\nu\log \tau)$ qubits. In Appendix~\ref{appendix:alternative f} we propose an alternative function $f$ which can be used for totally real $K$ and doesn't use the straddle encoding.
\end{defn}

Finally, one defined $f$ as follows:
\begin{equation}\label{eq:f}
f:G\xrightarrow{g}\{E\text{-ideals}\}\xrightarrow{|\str\rangle}\{\text{quantum states}\}=\H_Q^m.
\end{equation}

The following result is the conjunction of several results from \cite{eisentrager2014quantum}.
\begin{thm}[Theorems 5.5, D.4 and B.3 of \cite{eisentrager2014quantum}]\label{th:reduction}  Let $K$ be a number field of degree $n$, discriminant $D$, unit rank $m$ and regulator $R$. Let $L:=\Log \OO_K^*$ which is the hidden period of the function $f$ of Equation~\eqref{eq:f}.  

Set $s=3\cdot 2^{2n}\sqrt{nD}$ and $\nu=1/(4n(s\sqrt{n})^{2n})$.   
Then $f_5=\otimes^5 f$ is an 
$(r,a,\varepsilon)$-oracle with $\varepsilon=243/1024$,  $\Lip(f)=a = \frac{\sqrt{\pi n}s}{4\nu}+1 $ and $
r =   s\sqrt{n}^{n-1}2\nu \sqrt{m}$, where $c$ is an explicitly computable constant. In particular, $\varepsilon <1/4$, $r<\lambda_1(L)/6$ and
$$\log_2 \Lip(f) =O\left( m^2+m\log D\right).$$
\end{thm}
\begin{proof}
Note first that if $K=\Q$ or is quadratic imaginary then $m=0$ and there is nothing to be computed. Also note that $n/2-1\leq m\leq n-1$ so we can interchange $O(m)$ and $O(n)$ in the asymptotic complexity.

By~\cite[Th 5.5]{eisentrager2014quantum} $f$ is an $(r,\infty,\varepsilon)$-oracle hiding the lattice $L=\log \OO_K^*$ with $r=(s\sqrt{n})^{n-1}2\nu\sqrt{n}$ and $\varepsilon'=3/4$. 

By \cite[Th D.4]{eisentrager2014quantum} $f$ is $a$-Lipschitz with $a=\frac{\sqrt{\pi n}s}{4\nu}+1$. 

By \cite[Lem E.1]{eisentrager2014quantum} $f_5:=f\otimes f\otimes f\otimes f\otimes f$ is an $(5\Lip(f),r,(\varepsilon')^5)$-oracle hiding the same lattice.

In particular, $(\varepsilon')^5=243/1024<1/4$ satisfies the requirements of the CHSP solver (see Def~\ref{defn:dual lattice sampler}).

The condition $r< \lambda_1(L)/6$ is satisfied as $\lambda_1(L)\geq 1/2$ by \cite[Lem B.3]{eisentrager2014quantum} and $$r=O\left(\frac{s\sqrt{n}^{n-1}2}{4n(s\sqrt{n})^{2n}}\right)\leq \frac{1}{2s^{n-1}}\leq \frac{1}{6\cdot 2^n\sqrt{D}}  <\frac1{12}.$$

Finally, we have $$\log_2 \Lip(f)=  O(m s)=O(m^2+m\log D).$$

\end{proof}

\begin{corr}\label{cor:complexity}
Let $K$ be a number field of discriminant $D$ and unit rank $m$. For any error bound $\tau>0$. There exists a quantum algorithm of time $\mathrm{poly}(m,\log D,\log \tau)$ using a number of qubits
$$N_\text{qubits}=O(m^5+ m^4\log D)+O(m\log \tau^{-1})$$
which, for a set of units $\mu,\varepsilon_1,\ldots,\varepsilon_m$ such that $\mu$ is a root of unity and the other have infinite order and 
$$\OO_K^*\simeq \mu^{\Z/\omega}\times \varepsilon_1^\Z\times \cdots\times \varepsilon_m^\Z,$$
the algorithm outputs $\sigma_i(\log(\sigma_i(\varepsilon_j))_{1\leq i\leq m,1\leq j\leq n}+\tau_{i,j}$ with $\tau_{i,j}\in\R$ such that $|\tau_{i,j}|\leq \tau$.
\end{corr}
\begin{proof}
By Theorem~\ref{th:reduction}, there exists a function $f$ which hides $\Log \OO_K^*$ and which is an $(a,r,\varepsilon)$-oracle such that $r\leq \lambda_1(L)/6$ and $\varepsilon <1/4$. By Theorem~\ref{th:CHSP} there exists a polynomial-time algorithm which uses a number of qubits as in Equation~\eqref{eq:qubits}. 

In the CHSP solver one takes stores the values of $f$ on $Qm$ qubits, so the $O(N_f)=O(Qm)$. We are left with estimating $Q$. For this, the main part is Lemma~\ref{lem:dual}: $$O(\log (1/\lambda_1^*))= \log_2(\norm{B_L})=O(m+\frac1m\log D) .$$
We used the fact that $L$ admits an LLL-reduced basis, which is enough to bound $1/\lambda_1^*$, but this cannot be computed because we have no basis of $L$ until the end of the CHSP algorithm. We inject this in Equation~\eqref{eq:qubits} and obtain the announced value.

Given $\Log \OO_K^*$, we are left with computing $\mu,\varepsilon_1,\ldots,\varepsilon_m$. Let $\omega$ be the number of roots of unity of $K$ and recall that $\omega$ divides $D$. By multiplying the time by $\log D$, we can enumerate the divisors of $D$, so we can assume that we know $\omega$. 

 By Dirichlet's Theorem, an $m$-tuple such that $\Vol (\Log(\varepsilon_j))=R$ and a root of unity of order $\omega$, form a basis of the unit group.\footnote{We used here a different proof than in~\cite{eisentrager2014quantum} where one reduces the case of CHSP defined on arbitrary abelian groups to the case of CHSP over $\R^m$.} 
\end{proof}

\begin{rem}\footnote{We are indebted to Bill Allombert who has made this objection.}
The analytic class number formula states that $Rh=\sqrt{D}^{1+o(1)}$. Very little is known on the distribution of $h$ but no conjecture, e.g. the Cohen-Lenstra heuristic, is contradictory with the fact that $h=1$ and $R=\sqrt{D}^{1+o(1)}$ for a proportion of $1-o(1)$ of the number fields. Hence, $\max_i(\norm{\log \varepsilon_i})\sim R^{1/m}$ and the unit group is fully determined only when $\log \tau =\Omega(\sqrt{D})$. In that case CHSP doesn't compute the full unit group in polynomial-time. 

Note however that in the case of classical algorithms there is no algorithm which computes a partial information on the regulator without fully computing it (see for instance~\cite{Buchmann1988subexponential}). So the algorithm studied in this article suggests a quantum advantage in the case of unit groups.
\end{rem}

\begin{example}
\label{ex:complexity}
\begin{enumerate}
\item If $K=\Q(\zeta_n)$ with prime $n$ then $m=(n-1)/2$ and $\log_2D=(m-2)\log_2m$. If one applies Algorithm~\ref{algo:full L} without taking notes that it is a very particular, the algorithm uses $O(m^5\log m)$ qubits. 

\item If $K$ is a Kummer extension, i.e. $K=\Q(\sqrt[n]{D})$ for some integer without powers~$D$, then $O(\log \disc K)=O(\log D)$ and the number of qubits is $O(n^5+n^4\log D)$. The value of $n$ and $D$ are independent. When $n$ is fixed, i.e; the case of real quadratic fields, the number of qubits is $O(\disc K)$. When $n\sim \log D$, the number of qubits is $(\log D)^5/(\log\log D)^4$.   
\end{enumerate}
\end{example}

The large number qubits used by the present algorithm is in contrast with Shor's algorithm, which requires $2m+O(1)$ qubits in the case of factorization and discrete logarithms and $7m+O(1)$ in the case of elliptic curve discrete logarithms. The paragraph "Conclusion and Research Directions" at the end of~\cite[Sec 1]{boer2020quantum} states that, even if the complexity in $\log D$ is ignored and one uses approximation techniques in the quantum Fourier transform, the algorithm requires $\Omega(m^3)$ qubits. We refer to Appendix~\ref{appendix:security levels} for an informal discussion about the security levels one can propose in the case of quantum algorithms.

\subsection{Exploiting automorphisms}\label{sec:automorphisms}

In this section we tackle the case of cyclotomic fields without using the cyclotomic units. Instead we use the fact that $K$ has Galois automorphisms.

The main impact of this section is to reduce the practical cost of the Buchmann-Pohst step (which is non-quantum). Indeed, this step consists in applying LLL to a lattice $L$ which has more algebraic structure: it is a $\Z[\zeta_k]$-module for an integer~$k$. From an asymptotic point of view, Fieker and Stehlé~\cite{fieker2010short} proved that, if one is to find a new $\Z[\zeta_k]$-basis of $L$ which is shorter in a sense to be specified later, one has the best asymptotic complexity up to a polynomial factor if one follows the steps:
\begin{enumerate}
\item forget the $\Z[\zeta_k]$-module structure;
\item reduce the basis of the underlying $\Z$-lattice;
\item compute again the  $\Z[\zeta_k]$-module structure.
\end{enumerate}
From a practical point of view though, we shall explain that it is faster to work directly with the $\Z[\zeta_k]$-module structure.  

\paragraph{The structure of $\Z[G]$-lattice.} In the following, a lattice which has a structure of $R$-module for some ring $R$ is called a $R$-lattice. When $G:=\Aut(K/\Q)$, the group of Galois automorphisms, is abelian, as it is the case for cyclotomic fields, $\OO_K^*/( \OO_K^*)_\text{tors}$ has a structure of $\Z[G]$-lattice. Indeed, for $\sigma\in G$ and $u\in \OO_K^*$ we set $\sigma\cdot u=\sigma(u)$ and for set $\sigma_1,\ldots,\sigma_k\in G$ of ring generators of $\Z[G]$, and any $\lambda=\sum_{e_1,\ldots,e_k} c_{e_1,\ldots,e_k}\sigma_1^{e_1}\cdots \sigma_k^{e_k}$ we set
$$u^\lambda =\prod_{e_,\ldots,e_k} (\sigma_1^{e_1}(u)\cdots \sigma_k^{e_k}(u))^{c_{e_1,\ldots,e_k}}.$$

\begin{lem}
Let $G:=\Aut(F/\Q)$ and $\sigma\in G$ of order~$k$.
\begin{enumerate}
\item Then $\OO_F^*$ is an $\Z[\zeta_k]$-lattice. 
\item If $k$ is not of the form $$k=2^\varepsilon \prod p_i^{e_i}\text{ with }\varepsilon\in\{0,1\}\text{ and }p_i\equiv 1\pmod 4$$
then $\OO_F^*$ is a $\OO_{\Q(i\sqrt{d})}$-lattice for some divisor $d$ of~$k$.
\end{enumerate}
\end{lem}
\begin{proof}
1. If $\sigma\in G$, every $\Z[G]$-lattice is a $\Z[\sigma]$-lattice. In order to prove that $\Phi_k$ is the minimal polynomial of the endomorphism associated to $\sigma$, let $u\in \OO_F^*$ and set $v=u^{\Phi_k(\sigma)}$. Then $v^{\sigma-1}=u^{\sigma^k-1}=u/u=1$ and further $\sigma(v)=v$. But $\sigma\in \Aut(F/\Q)$, so $v\in \OO_F^*\bigcap \Q=\pm1\in (\OO_F^*)_\text{tors}$. Equivalently, $\Phi_k(\sigma)(u)\in (\OO_F^*)_\text{tors}$ and $\Phi_k$ is an annihilating polynomial of $\sigma$. Since it is irreducible, it is the minimal polynomial of $\sigma$. But $\zeta_k$ has the same minimal polynomial, so $\Z[\sigma]\simeq \Z[\zeta_k]$.\\ 
2. It is a classical result that, if one sets $p^*=(-1)^\frac{p-1}{2}p$ for a prime $p$, then  $$\Q(\sqrt{p^*})\subset \Q(\zeta_p).$$ 
When $k'\mid k $ one has $\Q(\zeta_{k'})\subset \Q(\zeta_k)$. Hence if $k$ is divisible by $4$ or a prime $p\equiv 3\pmod 4$, one can take $d=4$ or $d=p$ in the statement of the lemma. 
\end{proof}
\paragraph{Norm-Euclidean rings $\OO_K$.}Recall that $\Z[i]=\Z[\zeta_4]$ and $\OO_{\Q(i\sqrt{3})}= \Z[\zeta_3]$ are norm-Euclidean: let $k=3$ or $4$, in order to divide $a\in \Z[\zeta_k]$ by $b\in\Z[\zeta_k]$ one rounds each coordinate of $a/b\in \Q[\zeta_k]$ and denotes the result by $q$, and then one sets $r=a-bq$; if $N_{K/\Q}(r) < N_{K/\Q}(b)$ for all $a,b$ we say that $\Z[\zeta_K]$ is norm-Euclidean. Napias~\cite{Napias} showed that the LLL algorithm extends naturally to the norm-Euclidean rings $\Z[i]$ and $\Z[\zeta_3]$.

\begin{defn}[norm-Euclidean ring, Prop 1.2.5 of \cite{Camus2017}, Def 2 of \cite{kim2017lattice}] The Euclidean minimum of a number field $K$ is
$$\m_K:=\max_{x\in K}\min_{y\in \OO_K} N_{K/\Q}(x-y).$$
In particular, when $K$ is  imaginary quadratic, $ \m_K= \max_{x\in \C}\min_{y\in \OO_K} |x-y|^2$. If $\m_K<1$ we say that $K$ is norm-Euclidean. 
\end{defn}
\begin{example}[Prop 1.2.4 of~\cite{Camus2017},Prop 3 of~\cite{kim2017lattice}] The following fields are norm-Euclidean:
\begin{itemize}
\item $\Q(\zeta_k)$ for $k\in\{ 1, 3, 4, 5, 7, 8, 9, 11, 12, 15, 16, 20, 24\}$;
\item $\Q(i\sqrt{d})$ for $d\in\{ 1,2,3,7,11\}$.
\end{itemize}
\end{example}
\begin{defn}{($\OO_K$-LLL-reduced basis)}
 Let $K=\Q(i\sqrt{d})$ be an imaginary quadratic field. We identify $K$ with its embedding in $\C$ such that $\sqrt{-d}$ is mapped to a complex number of positive imaginary part.
 
 We consider the dot product $\langle\cdot,\cdot\rangle$ of $\C^m$ and we set
$\norm{x}=\sqrt{\langle x,x \rangle}$, which is a norm. 
 
  Let $b_1,b_2,\ldots,b_m$ be independent in $\C^m$. We say that $L:=\oplus_{j=1}^m \Z[\zeta_k]b_j$ is a $\Z[\zeta_k]$-lattice and $\{b_j\}$ is a basis. Let $\{b_1^*=b_1,b_2^*,\ldots,b_m^*\}$ be the Gram-Schmidt orthogonalization of $\{b_j\}$ with respect to the dot product. We set $\mu_{i,j}=\langle b_i,b_j^*\rangle/\norm{b_j^*}^2$. We say that $\{b_j\}$ is $\OO_K$-LLL-reduced with respect to $\m_K<\delta<1$ if the following two conditions hold:
\begin{enumerate}
\item $\norm{\mu_{i,j}}\leq \mathfrak{M}_K$ for $1\leq j<i\leq m$; \hfill(size reduced)
\item $ \norm{b_i^*}^2+\norm{\mu_{i,i-1}}^2\norm{b_ {i-1}^*}^2 \geq \delta  \norm{b_{i-1}^*}^2  $.\hfill (Lov{\'a}sz condition)
\end{enumerate}
\end{defn}

In two contemporary works Kim and Lee~\cite{kim2017lattice} and Camus~\cite{Camus2017} extended LLL to all norm-Euclidean rings $\Z[\zeta_k]$ and $\OO_{\Q(i\sqrt{D})}$. In particular, when $K$ is a norm-Euclidean quadratic imaginary field, \cite[Th. 1.3.8]{Camus2017} states that any $\OO_K$-reduced basis $b_1,\ldots,b_m$ of $L$ is such that
\begin{equation*}\label{eq:LLL}
\norm{b_j}\leq \left(\frac{1}{\delta-\m_K}\right)^{j-1} (\det L)^{1/m}.
\end{equation*}

A direct  application of $\OO_K$-LLL is the $\OO_K$ variant of Buchmann-Pohst algorithm, that we propose in Algorithm~\ref{algo:Buchmann-Pohst}.

\begin{algorithm}
\begin{algorithmic}[1]
\Require a $\OO_K$-lattice $L\subset \C^m$ given by approximations $\widetilde{g_1},\ldots,\widetilde{g_k}\in \C^m$ of  $g_1,\ldots,g_k \in \C^m$ which $\OO_K$-span $L$; a lower bound $\mu$ of $\lambda_1(L)$; an upper bound $D$  on $\det L$
\Ensure approximations of a basis $b_1,\ldots,b_m$ over $\OO_K$ of $L$

\State $B\gets c_K^m D^{1/m}$; $C\gets (B/\mu)^m\gamma_m^{1/2}$; $\tilde{M}\gets (k\sqrt{m}/2+\sqrt{k})C$; $q\gets \lceil\log_2\left( 
(\sqrt{mk}+2)\tilde{M}2^{(k-1)/2}/\mu
\right) \rceil$
\State the the $\OO_K$-LLL-reduction of the following matrix on the left to obtain the matrix on the right
\begin{small}
\begin{center}
$\left(
\begin{array}{ccc}
\begin{array}{|c|}
\hline 
 \\
\lfloor \widetilde{g_1}2^q \rfloor \\
 \\
\hline
\end{array}
&\cdots  &
\begin{array}{|c|}
\hline 
 \\
\lfloor \widetilde{g_k}2^q \rfloor \\
 \\
\hline
\end{array}
 \\
1 &  &  \\
  & \ddots & \\
  &        & 1 
\end{array}
\right)
\mathrel{\leadsto}
\left(
\begin{array}{ccc|ccc}
\begin{array}{|c|}
\hline 
 \\
\widetilde{c_1}\\
 \\
\hline
\end{array}
&\cdots  &
\begin{array}{|c|}
\hline 
 \\
\widetilde{c_k} \\
 \\
\hline
\end{array}
\\
\begin{array}{|c|}
\hline 
 \\
\textbf{m}_1 \\
 \\
\hline
\end{array}
& 
\cdots
&
\begin{array}{|c|}
\hline 
 \\
\textbf{m}_{k}\\
 \\
\hline
\end{array}
\end{array}
\right)$
\end{center}
\end{small}
\State \Return $\widetilde{c}_{k-m+1},\widetilde{c_2},\ldots,\widetilde{c_k}$
\end{algorithmic}
\caption{Buchmann-Pohst over $\OO_K$}
\label{algo:Buchmann-Pohst}
\end{algorithm}

\begin{thm}[adaptation to $\OO_K$-lattices of Th 3.1 in~\cite{BuchmannPohst}]\label{th:Buchmann-Pohst}
Algorithm~\ref{algo:Buchmann-Pohst} is correct and terminates in $O\left(  
(k+m)^6(m+\log(D^{1/m}/\mu))
\right)$ operations on a classical computer.
\end{thm}
\begin{proof}
The proof is a verbatim translation of the correctness proof in the case of $\Z$-lattices. By Equation~\ref{eq:LLL}, if $b_1,\ldots,b_m$ is a reduced basis of a lattice such that $\det L\leq D$, then $\max(\norm{b_1},\ldots,\norm{b_m})\leq B$. This is used in the proof of \cite[Prop 2.2]{BuchmannPohst} to show that the successive minima of the lattice generated by the columns of the matrix in Algorithm~\ref{algo:Buchmann-Pohst} are bounded by $\tilde{M}$. We use the upper bound of the norm of an $\OO_K$-reduced basis of~\cite[Prop 1.3.8]{Camus2017} instead of~\cite[(1.12)]{LLL82} to obtain the equivalent of~\cite[Eq (9)]{BuchmannPohst}:
\begin{equation}\label{eq:c_j upper}
\norm{\widetilde{c_j}}\leq 2^{(k-1)/2}\tilde{M}\qquad1\leq j\leq k-m.
\end{equation}
For any vector $(\widetilde{c},\textbf{m})\in \oplus_{j=1}^k\Z (\tilde{c}_j,\textbf{m}_j)$ such that $\textbf{m}=(\textbf{m}^{(1)},\ldots,\textbf{m}^{(k)})$ is not a relation of $b_1,\ldots,b_k$, i.e. $\textbf{m}\not\in L^*$, i.e. $\sum_{j=}^k \textbf{m}^{(j)} \widetilde{b}_j\neq 0.$
Then, the arguments in the proof of~\cite[Th 3.1]{BuchmannPohst} can be copied in a verbatim manner to obtain:
\begin{equation}\label{eq:c_j lower}
\norm{\widetilde{c}_j}> 2^{(k-1)/2}\tilde{M}\qquad 1\leq j\leq k.
\end{equation}
Since Equations~\eqref{eq:c_j upper}~and~\eqref{eq:c_j lower} are contradictory for $1\leq j\leq k-m$, the only possibility is that $\textbf{m} \in L^*$. We do the same transformations on the matrix with $b_j$ instead of $\widetilde{b_j}$ and call $c_j$ the vectors which replace $\widetilde{c_j}$ in this case. What we have proved is that $c_1=c_2=\cdot=c_{k-m}=0$.

Since the transformations done in the LLL algorithm on the columns of the matrix are reversible, they preserve the rank of any subset of rows of the matrix.  The rank of the first $m$ rows of this matrix is $\dim L=m$. Since $c_1=c_2=\cdot=c_{k-m}=0$, the vectors $c_{k-m+1},\ldots,c_k$ form a basis of~$L$.  
\end{proof}

\paragraph{LLL over $\Z[i\sqrt{d}]$ as a practical improvement} The constant hidden in the big Oh of implementation is (inverse) proportional to $\log \delta$, which (inverse) proportional to $\log (\delta-\m_K)$. Hence the difference is of only a few percentages. The dependence of the time in the rank  is quartic, so we replace a number of operations of high-precision real numbers by $2^4=16$ times less complex numbers at the same level of precision. Finally, using the  Karatsuba trick, a multiplication of complex numbers costs $3$ multiplications of real numbers, so the overall gain is a factor $16/3\approx 5.33$. An implementation~\cite{Elbaz2021} of $\Z[i\sqrt{d}]$-LLL shows that the reduction of an $\OO_K$-lattice is $\approx 5$ times faster when its algebraic structure is used when compared to forgetting it and using only the underlying structure of $\Z$-lattice: they used $\OO_K=\Z[i]$, the coordinates of the basis vectors have $512$ bits and $\delta=0.99$ both over $\Z$ and $\Z[i]$.\footnote{A similar speed-up was obtained in the case of $\Z[\zeta_k]$-LLL in~\cite{kim2017lattice}: an example took 20 s over $\Z[\zeta_k]$ compared to 75 s  over~$\Z$.}

\begin{paragraph}{Consequences in cryptography.}
Given our current contribution, the speed-up due to automorphisms is $5.33$. It is an open question to extend it to a larger class of $\OO_K$-rings (e.g.~\cite{PelletMary2019} investigate if LLL can be extended to arbitrary $\OO_K$ by solving CVP instances in dimension $[\OO_K:\Z]$.). If Buchmann-Pohst can be extended, the best speed-up that one can target is $[\OO_K:\Z]^3$. Indeed, the dependence of LLL in the dimension is quartic whereas the cost of the arithmetic over $\OO_K$ is quadratic or quasi-linear if fast arithmetic is used  

The speed-up due to the automorphisms  is  polynomial in the number of automorphisms. This is a familiar situation because a similar speed-up happens for automorphisms in the case of classical algorithms for factorization, discrete logarithms in finite fields and discrete logarithms on elliptic curves (see e. g.~\cite[Sec 4.3]{JLSV2006},\cite[Sec 5.3]{BGGM2015}). For instance, the attacks in~\cite{Morain1999} against ECDSA using endomorphisms of the curve achieved a polynomial speed-up in an algorithm of  exponential complexity. The community reacted by suggesting to use elliptic curves with endomorphisms as their effect was now considered to be benign, e.g.~\cite{GLV2001}. In a similar manner, if the security of the class group of a given degree and determinant is considered sufficient, we suggest to use fields with automorphisms in order to speed-up the protocol, e.g. it has been done by XTR in cryptosystems based on discrete logarithms in finite fields~\cite{XTR2000}. 
\end{paragraph}


\section{A new algorithm using cyclotomic units}\label{sec:cyclotomic}
When $K=\Q(\zeta_m)$ for an arbitrary integer $m$ we  define the group of cyclotomic units to be the subgroup $C$ of $K^*$ generated by $-1$, $\zeta_m$ and $\zeta_m^j-1$ with $j=1,2,\ldots,m-1$, intersected with the group of units~$\OO_K^*$. 
We follow the notations of~\cite[Sec. 3]{Cramer2021mildly}, in particular $m$ and $k$ don't have the same meaning as in the other sections. Factor  $m=p_1^{\alpha_1}p_2^{\alpha_2}\cdots p_k^{\alpha_k}$ and, for each index $i$, put $m_i=m/p_i^{\alpha_i}$.

For $j=1,\ldots,m-1$ we set
\begin{equation*}
v_j = \left\{ 
\begin{array}{ll}
1-\zeta_m^j,&\text{if for all $i$ we have $m_i\nmid j$}, \\
\frac{1-\zeta_m^j}{1-\zeta_m^{m_i}},&\text{otherwise for the unique $i$ such that $m_i\mid j$.}
\end{array}\right. 
\end{equation*}

\subsection{Unconditional results}

\begin{lem}[Theorem 4.2 in \cite{Kucera1992}]\label{lem:cyclotomic units}
The lattice $M:=\Log C$ admits the system of generators $\{b_j=\Log (v_j)\mid j=1,2,\ldots,m-1\}$. 
\end{lem}

\begin{lem}[Lemma 3.5 in \cite{Cramer2021mildly}]\label{lem:mildly}
For any integer $j\in\{1,2,\ldots,m-1\}$ we have $\norm{1-\zeta_m^j}_2=O(\sqrt{m})$. Hence $\norm{B_M}_2=O(\sqrt{m})$.
\end{lem}

\begin{thm}
Algorithm~\ref{algo:full L} computes a basis of the unit group of $\Q(\zeta_m)$ in $\mathrm{poly}(m)$ time and uses $O( m^2\log m)$ qubits.
\end{thm}
\begin{proof}
The only quantum step of Algorithm~\ref{algo:full L} is the dual lattice sampler, whose parameters are $\delta=\norm{B_L}/\lambda_1^*(L)$ and $ \eta=1/k^2$ where $$k = \log_2(\sqrt{m}\Lip(f)) +\log_2(\det L))$$ is given in step 1 of the algorithm. In this section $m$ is not necessarily equal to the rank $m'$ of $\Z[\zeta]^*$ but $O(m')=O(m)$ so that we use $m$ and $m'$ interchangeably.
By Lemma~\ref{lem:dual lattice sampler}, the sampler used $O(Qm)$ qubits with $Q$ as below. We write $a\ll b$ for $a=O(b)$. We use Th~\ref{th:reduction}: $\log \Lip(f)\ll m^2+m\log D$ and then $ k\ll m^2+m\log D$. We also use Lemma~\ref{lem:mildly}: $\log 1/(\delta\lambda_1^*(L))\ll \log (1/\norm{B_M}) \ll \log m$ and $\log D=(m-2)\log m\ll m\log m$.
$$
\begin{array}{rcl}
Q&\ll & m \log(m\log\frac{1}{\eta}))+\log\frac{1}{\delta\lambda_1^*(L)}+\log(\Lip(f))\\
&\ll & (m \log m +m\log \log k)+\log m+m^2+m\log D \\
&\ll &m\log m+m\log\log D+\log m+m^2+m\log D\\
&\ll &m^2+m\log D\\
&\ll & m^2\log m.
\end{array}
$$
\end{proof}
This is to be compared to the $O(m^5+m^4\log D)$ qubits used by number fields in general (Cor~\ref{cor:complexity}) and $O(m^5\log m)$ used by cyclotomic field if the Algorithm~\ref{algo:CHSP} were used without taking profit of the cyclotomic units (Ex~\ref{ex:complexity}).

\begin{rem}
Our improvement can be used whenever the lattice of units has a sublattice admitting a short basis and this is not limited to the cyclotomic fields and, more generally, abelian fields. For instance, Kihel~\cite{Kihel2001} proposed a family of fields with dihedral Galois group which have a full-rank  subgroup of units which are short. Although one expects a speed-up in all these cases, the cyclotomic case is even more special. Indeed, the parameter $\delta\leq \norm{B_M}/\lambda_1^*(L)$ depends in a large extent on the index $[L:M]$. By~\cite[Ex 8.5]{Washington} (see also~\cite[Th 2.8]{Cramer2016recovering}), in the case of cyclotomic fields of prime conductor, $ [L:M]=h^+(m)$. By Conjecture~\ref{conj:Weber} this is small, so the impact of the full-rank subgroup is more important in this case.  
\end{rem}

\subsection{Consequences of a recent conjecture for cyclotomic fields}\label{sec:HSP}
One denotes $h(m)$ the class group of $\Q(\zeta_m)$, $h^+(m)$ the class group of its maximal real subfield, i.e. $K:=\Q(\zeta_m+\zeta_m^{-1})$, and one sets $h^-(m)=h(m)/h^+(m)$. A folklore conjecture states that $h^-(m)$ is large but easy to compute whereas $h^+(m)$ is small but hard to compute. A recent work compiled existing conjectures and pushed further the numerical computations, so one can give a precise form to the folklore conjecture.
\begin{conjecture}[Assumption 2 in \cite{Cramer2021mildly}]\label{conj:Weber} For all integers $m$, 
$$h^+(m)\leq \mathrm{poly}(m)$$ for a fixed polynomial \textit{poly}.   
\end{conjecture}

\paragraph{Consequence of the conjecture : class number in real cubic fields without quantum computers}
Marie-Nicole Gras~\cite{Gras1975} proposed an algorithm to compute class numbers of cyclic cubic fields in polynomial time with respect to a bound on $h$. If $m$ is the conductor of a cyclic cubic field, then Conjecture~\ref{conj:Weber} implies that $h$ is polynomial in $m$.  
Note also that Schoof~\cite{Schoof2003} proposed an algorithm for $\Q(\zeta_m)^+$ with prime $m$ which is faster when a small bound on $h^+(m)$ is known.  
\paragraph{Reduction of unit group computation to HSP}

\begin{thm}
Under Conjecture~\ref{conj:Weber} the unit group of $K=\Q(\zeta_m)^+$ and its class group can be computed by an HSP algorithm in polynomial time and space. 
\end{thm}
\begin{proof}
Let $N=\mathrm{poly}(m)!$ and note that Conjecture~\ref{conj:Weber} implies that $\OO_K^*\subset C^{1/N}$. Let $m'=\varphi(m)/2-1$ and let $\varepsilon_1,\ldots,\varepsilon_{m'}$ be a basis of $C/\langle \zeta_m,-1\rangle$. Note that it can be computed using the HNF algorithm from the vectors $v_j$ of Lemma~\ref{lem:cyclotomic units}. Since $K$ is totally real, its only roots of unity are $ \pm 1$. For any $c\in C$, we call vectorization of $c$ and denote $\vect(c)$ the unique $(m'+1)$-tuple in $(e_0,e_1,\ldots,e_{m'})\in\Z/2\Z\times(\Z/N\Z)^{m'}$ such that 
$c=(-1)^{e_0}\prod_{j=1}^{m'}\varepsilon_j^{e_j}$. 

We set 
\begin{equation}
\begin{array}{rrcl}
f: & \Z/2\Z\times (\Z/N\Z)^{m'}  &\rightarrow & \left\{\text{
\begin{tabular}{c}
canonical representations\\
of Kummer extensions
\end{tabular}
}\right\} \\
   &(e_0,e_1,\ldots,e_{m'})&\mapsto &K(\zeta_N,\sqrt[N]{(-1)^{e_0}\prod_{j=1}^{m'} \varepsilon_j^{e_j}}))
\end{array}
\end{equation}
 We claim that the period of $f$ is 
 $$\mathrm{period}(f)=\{ \vect(u)\mid u \in (\OO_K)^* \}. $$
 Indeed, let $(e_j)$ and $(z_j)$ be such that $f(e_0,\ldots,e_{m'})=f(e_0+z_0,\ldots,e_m+z_{m'})$ and set $\varepsilon = (-1)^{e_0}\prod_j \varepsilon_j^{e_j}$ and $\zeta=(-1)^{z_0} \prod_j \varepsilon_j^{z_j}$. 
 
 We have $K(\zeta_N,\sqrt[N]{\varepsilon })=K(\zeta_N,\sqrt[N]{\varepsilon \zeta})$. Let $d:=[ K(\zeta_N,\sqrt[N]{\varepsilon }):K(\zeta_N)]$. Then the criterion of isomorphism of Kummer extensions~\cite[page 58]{Koch1992} states that $\sqrt[N]{\zeta}$ belongs to $K(\zeta_N)$. But $\sqrt[N]{\zeta}\in \R$ and the maximal real subfield of $K(\zeta_N)$ is $K$, so $\sqrt[N]{\zeta}\in \OO_K^*$. Conversely, if $\sqrt[N]{\zeta}\in \OO_K^*$ it is direct that $(z_0,\ldots,z_{m'})$ is a period of $f$. 

 Finally, to compute the complexity of the algorithm, we write $a=\poly(b)$ for $a=b^{O(1)}$ and $a=\polylog (b)$ for $a=(\log b)^{O(1)}$. By~\cite{kitaev1995quantum}, the cost of the HSP is $\polylog(N)$. Since, for any $n$, $n!\leq n^n$, we have 
 $$\text{time}=\polylog(N)=\poly(\mathrm{poly}(m)\log(\mathrm{poly}(m)) )=\poly(\poly(m))=\poly(m).$$
\end{proof}

\section{Conclusion and open questions}
\begin{enumerate}
\item The unit group algorithms follow a parallel path to those of other problems like the class group e.g Buchmann's algorithm computes the two groups together and the quantum algorithm for class groups~\cite{biasse2015quantum} is a generalization of the one of Hallgren et al~\cite{eisentrager2014quantum}. It is interesting to have a precise estimation of the number of qubits for class groups depending on the bound on the class number, on the computation of discrete logarithms in the class group etc.
\item The idea of a small dual basis can be generalized: a) to all abelian fields; b) to Galois fields with simple group of automorphisms and known units e.g.~\cite{Kihel2001}. 
\item The possibility of exploiting automorphisms of arbitrary order depends on the possibility of a Buchmann-Pohst algorithm for $\OO_K$ -lattices. One might explore such an algorithm which reduces the size of the norms on average but can locally increase them.
\item The quantum algorithms being probabilistic, the output is not certified. Can one do a classical algorithm to prove that a given set of units generate a subgroup which is $\ell$-saturated?
\item The technology of quantum computers is very new and it is not known how many physical qubits correspond to a given number of logical ones. Have the gates used by a HSP algorithms less errors and hence require less error-correction? A precise analysis goes beyond the scope of this article.
\end{enumerate}

\appendix

\section{An alternative function hiding the units}
\label{appendix:alternative f}

We place ourselves in the case where $K$ is totally real. As before we call $n$ its degree. Let $P$ be a polynomial which defines $K$ and let $\alpha_1,\ldots,\alpha_n$ be the $n$ roots of $P$ in $\R$. Let 
$$
\begin{array}{rccc}
w: & \{P\in \R[x] \mid \deg P\leq n-1\}& \rightarrow & \R^n \\
& P &\mapsto & (P(\alpha_1),\ldots,P(\alpha_n)).
\end{array}
$$
Clearly, $w$ is a linear isomorphism.

Assume that $\disc(f)$ is squarefree and then $\OO_K=\Z[x]/P$. 
We define $$
\begin{array}{cccc}
f:&\R^n&\rightarrow &(\R_n[x] \bmod \Z_n[x])\times (\R_n[x] \bmod \Z_n[x])\\
& (x_1,\ldots,x_n) & \mapsto & (w^{-1}(e^{x_1},\ldots,e^{x_n}), w^{-1}(e^{-x_1},\ldots,e^{-x_n})) \bmod \Z_n[x]^2.
\end{array}
$$ 
We claim that the set of periods of $f$ form a lattice $\Lambda$ such that 
$$2\Log \OO_K^* \subset \Lambda\subset  \Log \OO_K^*.$$
Indeed, let $\varepsilon$ be in $\OO_K^*$ and set $(x_1,\ldots,x_n)=\Log (\varepsilon^2)$.  Then $w^{-1}(e^{x_1,\ldots,e^{x_n})}$ is the representative of $\varepsilon^2$ in the normal basis of $K$. The coordinates are all integers because $\varepsilon^2$ is in $\OO_K=\Z[x]/\langle P\rangle$. Similarly, $w^{-1}(e^{-x_1},\ldots,e^{-x_n})$ is the representative of $\varepsilon^{-2}$ in the normal basis. Its coordinates are integers because $\varepsilon^{-2}\in _OO_K^*$.

Conversely, let $(x_1,\ldots,x_n)$ be such that $f(x_1,\ldots,x_n)=(0,\ldots,0)$ and let $\varepsilon=w^{-1}(e^{x_1},\ldots,e^{x_n})$. Since $f(x_1,\ldots,x_n)=0$, the coordinates of $\varepsilon$ in the normal basis are integers, so $\varepsilon\in \OO_K$. Similarly, $\varepsilon^{-1}\in \OO_K$, so $\varepsilon$ is a unit.

Finally, given $\Lambda$, one solves a linear system over $\Z/2\Z$ to find $\Log \OO_K^*$. Since $K$ is totally real, its roots of unity are $\pm 1$ and $\Log \OO_K^*$ completely determines $\OO_K^*$. 


\section{A discussion on quantum security levels}\label{appendix:security levels}

The NIST post-quantum challenge~\cite{NIST} is willingly open on the definition of the  computational resources on a quantum computer: 

\begin{minipage}{0.9\textwidth}
\textit{
``Here, computational resources may be measured using a variety of different metrics (e.g.,
number of classical elementary operations, quantum circuit size, etc.)''}
\end{minipage}

\noindent A study conducted by Mosca et al.~\cite{Mosca2021} revealed that part of the experts consider that on the medium term one should not consider attacks which uses billions of qubits but one should protect against attacks which use hundred qubits. A similar situation happens in the case of classical algorithms: in 2015 the recommendations of the standardisation agencies (NIST, ANSSI, etc) was to use 2048-bit RSA keys. However, in~\cite{Logjam2015} the authors made a precise estimation that RSA 1024 can be broken with their implementation of NFS on a million cores and they made a study that 98\% of a sample of million+ servers use RSA 1024. Moreover, more than $30\%$ were supporting RSA 768 which had been broken since 2009. 

It is then necessary to be more precise on the quantum resources. We propose a classification on the time complexity and the number of qubits, as in Table~\ref{table:security levels} (for a review of the quantum attacks on the various public-key primitives see e.g.~\cite{Revue2022}). In this light, a verifiable delay function~\cite{Wesolowski2019VDF} based on the class group of a number field of high-degree and  non-cyclotomic is more secure than~RSA.

\begin{table}
\begin{small}
\begin{center}
\begin{tabular}{c|c|c|c}
& polynomial & subexponential & exponential \\
& time & time & time\\
\hline
$O(m)$ qubits &
\begin{tabular}{c}
factorization\\ discrete log\\ EC discrete log\\ IQC
\end{tabular}
&  &
\begin{tabular}{c}
lattices\\ error correction codes
\end{tabular}
\\
\hline
$O(m^5\log m)$ qubits & high degree CGP & & \\
\hline
\begin{tabular}{c}
superpolynomial\\space
\end{tabular}& & isogenies &\\
\end{tabular}
\end{center}
\end{small}
\caption{Classification of cryptographic primitives w.r.t. the resources of quantum attacks: factorization, discrete logarithms, elliptic curve (EC) discrete logarithms, class group of orders of imaginary quadratic fields (IQC), supersingular isogenies, lattices, error correction cryptosystems and class group (CGP) computations.}
\label{table:security levels}
\end{table}

\bibliographystyle{plain}
\bibliography{refs}

\end{document}